\documentclass[11pt]{article}
\usepackage{amsfonts}
\usepackage{amsfonts}
\usepackage{amsfonts}
\usepackage{amsmath}
\usepackage{mathrsfs}
\usepackage{amsmath}
\usepackage{mathrsfs}
\usepackage{mathrsfs}
\usepackage{amsmath}
\usepackage{fancyhdr}
\usepackage{mathrsfs,epsfig,morefloats}
\usepackage{amsfonts}
\usepackage{natbib}
\usepackage{amsmath}
\usepackage{amsthm}
\usepackage{amssymb}
\usepackage{color}
\usepackage{lscape}
\usepackage{rotating}
\usepackage{caption2}
\usepackage{subfigure}
\usepackage{graphicx}
\graphicspath{ {C:/Users/jean/Desktop/NUS/} }

\usepackage{amssymb}

\newcommand{\Date}[1]{\def\@Date{#1}}
\def\today{\number\day~\ifcase\month\or
 January\or February\or March\or April\or May\or June\or
 July\or August\or September\or October\or November\or December\fi~\number\year}

\def\be{\begin{equation}}
\def\ee{\end{equation}}
\def\bea{\begin{eqnarray}}
\def\eea{\end{eqnarray}}
\def\bd{\begin{displaymath}}
\def\ed{\end{displaymath}}
\def\bda{\begin{eqnarray*}}
\def\eda{\end{eqnarray*}}
\def\bsm{\begin{small}}
\def\esm{\end{small}}

\def\t0{\theta_0}

\def\ha1{\hat \beta_1}

\def\bnt{\begin{enumerate}}
\def\ent{\end{enumerate}}
\def\T{{ \mathrm{\scriptscriptstyle T} }}

\def\AS{A\"{\i}t-Sahalia}

\def\bsc{\begin{scriptsize}}
\def\esc{\end{scriptsize}}

\newtheorem{theorem}{Theorem}

\newtheorem{lemma}{Lemma}
\newtheorem{cy}{Corollary}

\theoremstyle{definition}

\newtheorem{remark}{Remark}

\newcommand\independent{\protect\mathpalette{\protect\independenT}{\perp}}
\def\independenT#1#2{\mathrel{\rlap{$#1#2$}\mkern2mu{#1#2}}}

\makeatletter
\newcommand{\figcaption}{\def\@captype{figure}\caption}
\newcommand{\tabcaption}{\def\@captype{table}\caption}
\makeatother

\newcommand{\cov}{{\rm Cov}}
\newcommand{\diag}{{\rm diag}}

\newcommand{\ie}{\mbox{\sl i.e.\;}}

\newcommand{\tr}{\mbox{tr}}
\newcommand{\var}{\mbox{Var}}

\newcommand{\bA}{{\mathbf A}}
\newcommand{\bB}{{\mathbf B}}

\newcommand{\bI}{{\mathbf I}}

\newcommand{\bM}{{\mathbf M}}
\newcommand{\bQ}{{\mathbf Q}}
\newcommand{\bP}{{\mathbf P}}

\newcommand{\bU}{{\mathbf U}}
\newcommand{\bV}{{\mathbf V}}
\newcommand{\bW}{{\mathbf W}}
\newcommand{\bX}{{\mathbf X}}
\newcommand{\bY}{{\mathbf Y}}
\newcommand{\bZ}{{\mathbf Z}}
\newcommand{\ba}{{\mathbf a}}
\newcommand{\bb}{{\mathbf b}}

\newcommand{\bg}{{\mathbf g}}

\newcommand{\bm}{{\mathbf m}}

\newcommand{\balpha} {\boldsymbol{\alpha}}
\newcommand{\bbeta}  {\boldsymbol{\beta}}
\newcommand{\bfeta}  {\boldsymbol{\eta}}

\newcommand{\bepsilonb}{\boldsymbol{\varepsilon}}
\newcommand{\bepsilonbb}{\boldsymbol{\epsilon}}
\newcommand{\bOmega}{\boldsymbol{\Omega}}
\newcommand{\bomega}{\boldsymbol{\omega}}

\newcommand{\bSigma}{\boldsymbol{\Sigma}}

\newcommand{\bgamma}{\boldsymbol{\gamma}}

\newcommand{\bmu} {\boldsymbol{\mu}}
\newcommand{\bzeta} {\boldsymbol{\zeta}}

\newcommand{\bLambda} {\boldsymbol{\Lambda}}

\def\JRSSB{{\sl Journal of the Royal Statistical Society}, {\bf B}}

\def\BKA{{\sl Biometrika}}
\def\JASA{{\sl Journal of the American Statistical Association}}

\def\AS{{\sl The Annals of Statistics}}

\def\AP{{\sl The Annals of Probability}}

\usepackage{hyperref} 

\begin{document}

\title{ \bf High Dimensional Elliptical Sliced Inverse Regression in non-Gaussian Distributions}

\author{Jia Zhang\\Department of Statistics, Southwestern University of Finance and Economics\\Chengdu {\rm 611130}, China; \\jeanzhang9@2015.swufe.edu.cn \and Xin Chen\footnote{Corresponding author.}\\Department of Statistics and Applied Probability, National University of Singapore\\Singapore {\rm 117546}, Singapore; \\stacx@nus.edu.sg \and Wang Zhou\\Department of Statistics and Applied Probability, National University of Singapore\\Singapore {\rm 117546}, Singapore; \\stazw@nus.edu.sg }



\maketitle

\begin{abstract}
Sliced inverse regression (SIR) is the most widely-used sufficient dimension reduction method due to its simplicity, generality and computational efficiency. However, when the distribution of the covariates deviates from the multivariate normal distribution, the estimation efficiency of SIR is rather low. In this paper, we propose a robust alternative to SIR - called elliptical sliced inverse regression (ESIR) for analysing high dimensional, elliptically distributed data. There are wide range of applications of the elliptically distributed data, especially in finance and economics where the distribution of the data is often heavy-tailed. To tackle the heavy-tailed elliptically distributed covariates, we novelly utilize the multivariate Kendall's tau matrix in a framework of so-called generalized eigenvector problem for sufficient dimension reduction. Methodologically, we present a practical algorithm for our method. Theoretically, we investigate the asymptotic behavior of the ESIR estimator and obtain the corresponding convergence rate under high dimensional setting. Quantities of simulation results show that ESIR significantly improves the estimation efficiency in heavy-tailed scenarios. A stock exchange data analysis also demonstrates the effectiveness of our method. Moreover, ESIR can be easily extended to most other sufficient dimension reduction methods.

\end{abstract}

\begin{quote}
\noindent
{\sl Keywords}: Multivariate Kendall's tau; Sliced inverse regression; Elliptical distribution; Convergence rate; Central subspace; Sufficient dimension reduction. \end{quote}

\begin{quote}
\noindent
{\sl MSC2010 subject classifications}: Primary 62H86; secondary 62G20
\end{quote}

\thispagestyle{empty}
\pagenumbering{gobble}

\newpage
\pagenumbering{arabic}

\setcounter{page}{1}

\section{Introduction}
\label{s1}

In the regression model, let $Y \in \mathbb{R}$ denote the response variable and $\bX \in \mathbb{R}^p$ denote the covariates. If there exist orthogonal $1 \times p$ vectors $\bbeta_1, \dots, \bbeta_K$ with unit norm such that
\be \label{1.1}
Y \independent \bX | (\bbeta_1 \bX, \dots, \bbeta_K \bX), \quad (K \le p)
\ee
where $\independent$ represents independence, the column space of the $p \times K$ matrix $\bB = (\bbeta_1^{\T}, \dots, \bbeta_K^{\T})$ is defined as a dimension reduction subspace by \cite{Cook1994} and \cite{Cook1998}. The intersection of all the dimension reduction subspaces is still a dimension reduction subspace and is called the central subspace (\cite{Cook1994}, \cite{Cook1996}). Various methods have been proposed to estimate the central subspace in the literature, which are together referred as sufficient dimension reduction methods. Among them, sliced inverse regression (SIR, \cite{Li1991}) is the earliest and most popular method owning to its simplicity, generality and efficiency for computation. \cite{Li1991} proved the consistency of SIR for fixed $p$ setting. \cite{HsingCarroll1992} considered the case where each slice only contained two data points and gave the asymptotic normality results for the SIR estimator. Following their work, \cite{ZhuNg1995} derived the asymptotic properties of the sliced estimator for general cases. \cite{ZhuFang1996} proposed another version of SIR based on the kernel technique and obtained its asymptotic results. All the results summarized above are constricted to fixed $p$ context. \cite{ZhuMiaoPeng2006} studied the asymptotic behaviors of the SIR for $p$ diverging with $n$. A recent work given by \cite{Lin2017} treated the asymptotic performance of the SIR estimator from a different perspective. Instead of SIR, other methods designed for the estimation of the central subspace have also been investigated, including but not limited to sliced average variance estimator (SAVE) proposed by \cite{CookWeisberg1991} and \cite{Cook2000}, principal Hessian directions (PHD, \cite{Li1992}, \cite{Cook1998}), parametric inverse regression (PIR) suggested by \cite{BuraCook2001a} and \cite{BuraCook2001b}, minimum average variance estimator (MAVE, \cite{Xia2002}), contour regression (\cite{Li2005}), inverse regression estimator (\cite{CookNi2005}), the hybrid methods which combined SIR and SAVE in a convex way (\cite{ZhuOhtakiLi2006}), principal fitted components (\cite{Cook2007}), directional reduction (DR, \cite{LiWang2007}), likelihood acquired directions (\cite{CookForzani2009}), semiparametric dimension reduction methods (\cite{MaZhu2012}), and direction estimation via distance covariance (\cite{ShengYin2013}, \cite{ShengYin2016}).

Due to the simplicity and computational efficiency of its algorithm, SIR is the most widely used and most studied method in practice and in the literature. However, although it only requires the linearity condition (\cite{Li1991}) to have the consistency of the SIR estimator, 
the SIR performs much worse when the distribution of $\bX$ deviates from the normal case. This phenomenon can be seen quite clearly from our simulations. The fact is that the more the $\bX$ deviates form the multivariate normal distribution, the worse the performance of the SIR gets. In principal component analysis (PCA), which is a unsupervised version of dimension reduction, the deviation from normal assumption is also a serious problem. That is, this kind of deviation may lead to the PCA's inconsistency (\cite{JohnstoneLu2009}, \cite{HanLiu2016}) when the dimension $p$ of $\bX$ is growing with the sample size $n$. Aware of this inconsistency problem, \cite{HanLiu2016} proposed a new kind of PCA method based on the Kendall's tau matrix for elliptically distributed $\bX$, called elliptical component analysis (ECA). They proved that the ECA method is consistent for both sparse and non-sparse settings. In this paper, we novelly extend their idea from unsupervised learning to supervised learning via a generalized eigenvector problem (\cite{Li2007}, \cite{ChenZouCook2010}). Consequently, our method can address the problem of low efficiency of SIR in non-normal settings. Furthermore, our method is theoretically sound since the elliptical distribution family naturally satisfies the so-called linearity condition (\cite{Li1991}) and the merits of the introduction of the Kendall's tau matrix for elliptically distributed $\bX$ are then well kept in the sufficient dimension reduction.

The essential reason why we are concerned about the elliptical family comes from the wide range of application of the elliptically distributed data, especially in finance and economics where the distribution of the data often possesses high kurtosis and heavy tailed pattern. For example, \cite{HanLiu2016} studied a high dimensional non-Gaussian heavy-tailed data set on functional magnetic resonance imaging in their paper. \cite{Fan2015} considered the problem of covariance matrix estimation based on large factor model for elliptical data. The proof of the consistency of the SIR in \cite{Lin2017} was based on the assumption that $\bX$ follows a sub-Gaussian distribution. In this paper, we go further steps to investigate the elliptical family of the covariates. To tackle the heavy-tailed problem, we propose a new SIR method called elliptical sliced inverse regression estimator (ESIR) and study both its basic properties and high dimensional properties.

Note that \cite{Li1991} had a remark for the robust versions of SIR (Remark 4.4). \cite{Li1991} did not think this issue was crucial and suggested the influential design points be down-weighted or be screened out in the observational study. However, things are different in our paper where the focus is on the elliptical distributions with heavy tails. It's not a problem of experimental design, because the data points are not under control. Besides, screening out those ``bad'' points seems not appropriate. On one hand, the number of those ``bad'' points can be very large due to the heavy tails and removing them from the sample would worsen the estimation efficiency. On the other hand, heavy tails of the data are exactly what we care about (especially in finance and economics) and ignoring this feature could lead to misleading conclusion. To sum up, we believe that it's of great importance to do some in-deep research to address the common heavy-tailed issue.

In the second part of the paper, we construct the convergence rate of the ESIR estimator under the high dimensional setting. Specifically, we allow all of the dimension of the covariates $p$, the number of the slices $H$ and the number of the data points $l$ in each slice to grow with the sample size $n$ at a proper rate. This kind of study is of vital importance due to the escalating of computing power bringing us a large quantity of high dimensional data sets in various fields, as pointed out by \cite{ZhuMiaoPeng2006}.

The rest of the article is organized as follows. In the next section, some background knowledge is given about the elliptical distribution and the Kendall's tau matrix. In Section \ref{s3}, we propose the ESIR estimator, its basic properties and the ESIR algorithm. Consistency and convergence rate of the ESIR estimator for high dimensional $\bX$ are demonstrated in Section \ref{s4}. We present large numbers of  simulation results to compare the estimation efficiency of ESIR with the original SIR and to investigate the influence of $p$, $H$ and $n$ on the estimation accuracy in Section \ref{s5}. Section \ref{s6} concludes the paper and the last section reports the technical proofs of the theorems.

\section{Background}
\label{s2}

\subsection{Elliptical distribution}

Let $\bmu \in \mathbb{R}^p$ and $\bSigma \in \mathbb{R}^{p \times p}$ with full rank (we only consider the full rank case in this paper). If
\[
\bX \overset{\rm{d}}{=}\bmu + \xi \bA \bU
\]
where $\xi$ is a nonnegative continuous scalar random variable, $\bA \in \mathbb{R}^{p \times p}$ is a deterministic matrix with $\bA\bA^{\T}=\bSigma$ and $\bU \in \mathbb{R}^{p \times 1}$ is a uniform random vector on the unit sphere, we say $\bX$ follows an elliptical distribution, \ie, $\bX \sim EC_{p}(\bmu, \bSigma, \xi)$. Here, $\bX\overset{\rm{d}}{=}\bY$ means that the random vectors $\bX$ and $\bY$ follow the same distribution. Throughout the article, without loss of generality we assume $\mathbb{E}(\xi^2) = p$ to guarantee $\cov (\bX) = \bSigma$. The marginal and conditional distributions of an elliptical distribution still belong to the elliptical family and the independent sum of elliptical distributions is also elliptically distributed. Special cases of elliptical distribution include multivariate normal distribution, multivariate t-distribution, symmetric multivariate stable distribution, symmetric multivariate Laplace distribution and multivariate logistic distribution, etc.

Compared with the Gaussian or sub-Gaussian family, the elliptical family enables us to model complex data more flexibly. First of all, the elliptical family includes kinds of heavy-tailed distributions, while the Gaussian is characterized with exponential tail bounds. What's more, we can use elliptical family to describe tail dependence between variables (\cite{HultLindskog2002}, \cite{HanLiu2016}). Thus, elliptical family can be used to model complex data sets such as the financial data, genomic data and bio-imaging data and so on.

\setcounter{equation}{0}
\subsection{ Multivariate Kendall's tau}

Let $\widetilde{\bX}$ be an independent copy of a random vector $\bX \sim EC_{p}(\bmu, \bSigma, \xi)$. We introduce the population multivariate Kendall's tau matrix $\bM \in \mathbb{R}^{p \times p}$ (\cite{HanLiu2016}):
\[
\bM:=\mathbb{E}( \frac{(\bX-\widetilde{\bX})(\bX-\widetilde{\bX})^{\T}}{\|\bX-\widetilde{\bX}\|_2^2}).
\]
Let $\{\bX_{i}\}_{i=1}^n$ be $n$ independent replicates of $\bX$. The sample version of the Kendall's tau matrix is defined as
\[
\widehat \bM := \frac{2}{n(n-1)} \sum_{i'<i} \frac{(\bX_i-\bX_{i'})(\bX_i-\bX_{i'})^{\T}}{\|\bX_i-\bX_{i'}\|_2^2}.
\]
It is easy to derive that $\mathbb{E}(\widehat \bM) = \bM$, $\tr(\widehat\bM)=\tr(\bM)=1$ and $\widehat \bM$ and $\bM$ are both positive definite.
The sample Kendall's tau matrix is a second-order U-statistic with good properties, that is, the spectral norm of the kernel of the U-statistic
\[
k(\bX_i, \bX_{i'}) := \frac{(\bX_i-\bX_{i'})(\bX_i-\bX_{i'})^{\T}}{\|\bX_i-\bX_{i'}\|_2^2}
\]
is bounded by $1$ which makes $\widehat\bM$ enjoy some good theoretical properties. Furthermore, the convergence of $\widehat\bM$ to $\bM$ doesn't rely on the generating variable $\xi$ thanks to the distribution free property of the kernel(\cite{HanLiu2016}).

Although the multivariate Kendall's tau matrix is not identical or proportional to the covariance matrix $\bSigma$ of $\bX$, under some conditions they share the same eigenspace, see \cite{Marden1999}, \cite{Crous2002}, \cite{Oja2010} and \cite{HanLiu2016}. Moreover, by simple calculation we find that $\widehat \bM$ is a weighted version of the sample covariance matrix $\widehat \bSigma$, that is
\be \label{2.1}
\begin{split}
\widehat \bM &~ = \frac{1}{n(n-1)} \sum_{i'<i} \frac{2}{\|\bX_i-\bX_{i'}\|_2^2} (\bX_i-\bX_{i'})(\bX_i-\bX_{i'})^{\T}\\
&~ := \frac{1}{n(n-1)} \sum_{i'<i} \omega_{ii'} (\bX_i-\bX_{i'})(\bX_i-\bX_{i'})^{\T},
\end{split}
\ee
while
\[
\widehat \bSigma = \frac{1}{n(n-1)} \sum_{i'<i} (\bX_i-\bX_{i'})(\bX_i-\bX_{i'})^{\T}.
\]

\setcounter{equation}{0}
\section{Elliptical sliced inverse regression}
\label{s3}

\subsection{Sliced inverse regression}

In this section, we give a rough overview of the SIR method. The model below is used to explore the theoretical properties of the SIR:
\be \label{3.1}
Y = f (\bbeta_{1} \bX, \dots, \bbeta_{K} \bX, \bepsilonbb),
\ee
where $\bbeta_1, \dots, \bbeta_K$ are unknown $p$ dimensional row vectors, $\bepsilonbb$ is independent of the covariates $\bX$ and $f$ is an arbitrary unknown function defined on $\mathbb{R}^{K+1}$. The linear space $\bB$ generated by $\bbeta$'s is called the efficient dimension reduction (e.d.r.) space and any linear combination of the $\bbeta$'s is referred as an e.d.r. direction. \cite{Li1991} demonstrated that if $\bX$ was standardized by $\bSigma = \cov (\bX)$ to have zero mean and identity covariance matrix, the inverse regression curve $\mathbb{E}(\bX|Y)$ would be contained in the e.d.r. space. Accordingly, the principal component analysis method is applied to the estimated covariance matrix of the inverse regression curve. Hence, the leading vectors of the estimated covariance matrix can then be transformed to estimate the e.d.r. directions. In fact, \cite{Li1991} showed that each $\hat\bbeta_k$ would converge to an e.d.r. direction at rate of $n^{-1/2}$ when $p$ stayed fixed.

The key condition for the SIR method is referred as the linearity condition, \ie, for any $\bb \in \mathbb{R}^p$, $\mathbb{E}(\bb^{\T}\bX | \bbeta_1, \dots, \bbeta_K) = c_0 + c_1 \bbeta_{1} \bX + \dots + c_K \bbeta_K \bX$ for some constants $c_0, \dots, c_K$. To satisfy this condition, the distribution of the covariates is required to be elliptically symmetric. Such distributions include the normal distribution and the general symmetric elliptical distributions.

\subsection{Elliptical sliced inverse regression}

We construct our basic theorem for elliptical sliced inverse regression (ESIR) in this part. Here, ``E'' represents our focus on the heavy-tailed symmetric elliptical family.

\begin{theorem}\label{t1}
Assume that $\bX \sim EC_p(\bmu, \bSigma, \xi)$. Under (\ref{3.1}) and the linearity condition given by \cite{Li1991}, the curve $\mathbb{E}(\bX | Y) - \mathbb{E}(\bX)$ is contained in the linear subspace spanned by $\bbeta_k \bM (k = 1, 2, \dots, K)$, where $\bM$ denotes the Kendall's tau matrix of $\bX$.
\end{theorem}

\begin{proof}
Following the conclusion of Theorem 3.1 in \cite{Li1991}, if we can prove that the linear space spanned by $\bbeta_k \bSigma(k=1, \dots, K)$ is the same as the space spanned by $\bbeta_k \bM (k=1, \dots, K)$, we are done.

For any vector $\bgamma \in \mathbb{R}^{K}$, let $\bB = (\bbeta_1, \dots, \bbeta_k)^{\T}$, then the span of $\bbeta_k \bM (k=1, \dots, K)$ can be written as $\bgamma^{\T}\bB\bM$ and
\[
\begin{split}
\bgamma^{\T}\bB\bM
&~ = \bgamma^{\T}\bB \cdot (\sum_{j=1}^{p} \lambda_j(\bM)\bmu_j(\bM)\bmu_{j}^{\T}(\bM))\\
&~ = \bgamma^{\T}\bB \cdot (\sum_{j=1}^{p} \lambda_j(\bM)\bmu_j(\bSigma)\bmu_{j}^{\T}(\bSigma))\\
&~ = \bgamma^{\T}\bB \cdot (\sum_{j=1}^{p} \mathbb{E}(\frac{\lambda_j(\bSigma) Q_j^2}{\lambda_1(\bSigma) Q_1^2 + \dots + \lambda_p(\bSigma)Q_p^2}) \bmu_j(\bSigma)\bmu_{j}^{\T}(\bSigma))\\
&~ := s \bgamma^{\T}\bB \cdot (\sum_{j=1}^{p} \lambda_j(\bSigma)\bmu_j(\bSigma)\bmu_{j}^{\T}(\bSigma))\\
&~ = (s\bgamma)^{\T}\bB \bSigma
\end{split}
\]
where the first equality comes from the spectral decomposition of $ \bM$, the second one is established by the property that $\bmu_j(\bSigma)= \bmu_j(\bM)$ (see \cite{Marden1999}, \cite{Crous2002} ,\cite{Oja2010} and \cite{HanLiu2016} for details), the third equality is given by Proposition 2.1 of \cite{HanLiu2016} and $\bQ := (Q_1, \dots, Q_p)^{\T} \sim N_p(\mathbf{0}, \bI_p)$. The last equality verifies our guess.
\end{proof}

Theorem \ref{t1} is very similar to Theorem 3.1 of \cite{Li1991}. It's not surprising in view of the close relationship between $\bSigma$ and $\bM$ in Section \ref{s2}. Worth noting that we only assume that $\bX$ follows an elliptical distribution and we don't pose any distribution restriction on $\bX|Y$ while \cite{BuraForzani2015} and \cite{Bura2016} required $\bX|Y$ be elliptically distributed and multivariate exponentially distributed respectively.

Let $\bZ = \bM^{-1/2} [\bX - \mathbb{E}(\bX)]$, where $\bM$ is the population Kendall's tau matrix defined above in Section \ref{s2}. Then, (\ref{3.1}) can be rewritten as
\be \label{3.2}
Y = f(\bfeta_1 \bZ, \dots, \bfeta_K \bZ, \bepsilonbb),
\ee
where $\bfeta_k = \bbeta_k \bM^{1/2} (k=1, \dots, K)$. Following the usage in \cite{Li1991}, we call the vector linearly generated by these $\bfeta_k$'s as the standardized e.d.r direction. For this new version of standardized covariates, we have the following corollary.

\begin{cy}\label{c1}
Under (\ref{3.2}) and the linearity condition given by \cite{Li1991}, the curve $\mathbb{E}(\bZ|Y)$ is contained in the linear space generated by $\bfeta_k$'s defined in (\ref{3.2}).
\end{cy}

We then can easily derive that $\cov[\mathbb{E}(\bZ|Y)$ is degenerate in directions which are orthogonal to $\bfeta_k$'s. Thus, by Lemma \ref{l1} in the Appendix, we conclude that
\be \label{3.3}
\bM_{\mathbb{E}(\bZ|Y)} = \mathbb{E}( \frac{(\mathbb{E}(\bZ|Y)-\mathbb{E}(\widetilde\bZ|\widetilde{Y}))(\mathbb{E}(\bZ|Y)-\mathbb{E}(\widetilde\bZ|\widetilde{Y}))^{\T}}
{\|(\mathbb{E}(\bZ|Y)-\mathbb{E}(\widetilde\bZ|\widetilde{Y}))\|_2^2})
\ee
is also degenerate in those directions mentioned above. Thus, the eigenvectors associated with the largest $K$ eigenvalues of $\bM_{\mathbb{E}(\bZ|Y)}$ are the standardized e.d.r. directions $\bfeta_k (k=1, \dots, K)$. We then transform $\bfeta_k (k=1, \dots, K)$ to $\bbeta_k$ by $\bbeta_k = \bfeta_k \bM^{-1/2} (k=1, \dots, K)$ for the original e.d.r. directions.

Given Corollary \ref{c1} and Lemma \ref{l1} in the Appendix, we construct the operating scheme for elliptical sliced inverse regression:

1. For each $\bX_i (i=1, 2, \dots, n)$, we calculate a new standardized form of $\bX_i$: $\widetilde \bX_i = \widehat \bM^{-1/2} (\bX_i - \bar{\bX})(i=1, 2, \dots, n)$, where $\widehat \bM$ and $\bar{\bX}$ denote the sample Kendall's tau matrix and the sample mean of $\bX$ respectively.

2. Divide the range of $Y$ into $H$ ``equal'' slices, $I_1, \dots, I_H$. Here, ``equal'' means that the number of the data points falling in each slice is equal to $l = \lfloor \frac{n}{H} \rfloor$.

3. In each slice, compute the sample mean of $\widetilde\bX$: $\hat{\bm}_h = 1/l \sum_{\mathrm{y}_i \in I_h} \widetilde{\bX}_i (h = 1, \dots, H)$.

4. Form the Kendall's tau matrix for $\hat{\bm}_h$:
\be \label{3.4}
\widehat{\bM}_{\bm} = 2/(H(H-1)) \sum_{h'<h} \frac{(\hat{\bm}_h-\hat{\bm}_{h'})(\hat{\bm}_h-\hat{\bm}_{h'})^{\T}}{\|\hat{\bm}_h-\hat{\bm}_{h'}\|_2^2} \quad (h = 1, \dots, H),
\ee
then compute the eigenvalues and eigenvectors of $\widehat{\bM}_{\bm}$.

5. Denote the largest $K$ eigenvectors of $\widehat{\bM}_{\bm}$ be $\hat{\bfeta}_k (k=1, \dots, K)$. We transform them back to the original version by $\widehat{\bM}$, that is, $\hat{\bbeta}_k = \hat{\bfeta}_k \widehat{\bM}^{-1/2}$.

In this algorithm, we enforce the number of the data points in each slice to be fixed to $l$ so that we don't need any weighting adjustment for the calculation of $\widehat\bM_{\bm}$. In addition, the data points in the last slice may be not exactly $l$ which exerts little influence asymptotically. Our algorithm is essentially a generalized eigenvector problem. Please see more details in \cite{Li2007} and \cite{ChenZouCook2010}. Of note, like SIR, the ESIR may not recover all the e.d.r. directions. One may refer to other dimension reduction methods like SAVE and DR to address such problems.

\setcounter{equation}{0}
\section{Asymptotic properties of elliptical slice inverse regression with diverging number of covariates}
\label{s4}

In this paper, we assume the data points $l$ stay the same in different slice, and the number of the slices $H$ and $l$ are both allowed to grow with the sample size $n$. During the process of the proof, we use original $\bX$ rather than its standardization. The conclusion can then be directly extended to the standardized version.

Denote the inverse regression curve by $\bm (Y) = \mathbb{E}(\bX|Y)$ and decompose $\bX$ as:
\[
\bX = \bm (Y) + \bepsilonb,
\]
where $\bm (Y) = (m_1(Y), \dots, m_p(Y))^{\T}$ with $m_{i}(Y) = \mathbb{E} (X_i | Y)$ and $\bX = (X_1, \dots, X_p)$  and for the sample version,
\[
\bX_i = \bm (Y_i) + \bepsilonb_i = \bm_i + \bepsilonb_i, \qquad i = 1, \dots, n
\]
and
\[
\bX_{(i)} = \bm (Y_{(i)}) + \bepsilonb_{(i)} = \bm_{(i)} + \bepsilonb_{(i)}, \qquad i = 1, \dots, n
\]
where $Y_{(1}) \le \dots, \le Y_{(n})$ and $\bX_{(i)}$ and $\bepsilonb_{(i)}$ are the concomitants (\cite{Yang1977}) of $\bm (Y_{(i)})$. For each slice, denote
\[
\bX_{(hi)} = \bm (Y_{(hi)}) + \bepsilonb_{(hi)} = \bm_{(hi)} + \bepsilonb_{(hi)}, \qquad i = 1, \dots, l, \quad h = 1, \dots, H.
\]
Here, $\bX_{(hi)} = \bX_{(l(h-1)+i)}$ and $Y_{(hi)} = Y_{(l(h-1)+i)}$. Of note, all the notations given above depend on $n$. Under several mild conditions, we establish the consistency and convergence rate of the ESIR estimator in the following theorem.

\begin{theorem}\label{t2}
Assume the following conditions hold:

 \rm{(1)}. $\bX \sim EC_p(\bmu, \bSigma, \xi)$ and $\sup_{i \le p} \mathbb{E} (|x_i|^{m}) \le \infty $ for some constant $m \ge 2$, where $\bX = (x_1, \dots, x_p)^{\T}$.

 (2). There exist some positive constants $C_1$ and $C_2$ such that
 \[
 C_1 \le \lambda_{\min}(\bM_{\mathbb{E}(\bX|Y)}) \le \lambda_{\max} (\bM_{\mathbb{E}(\bX|Y)}) \le C_2.
 \]

 (3). The inverse regression curve $\bm(Y) = \mathbb{E}(\bX|Y)$ is $\vartheta$-stable regarding $Y$ and $\bm(Y)$. See more details in \cite{Lin2017}.


Let $\{\bX_j\}_{j=1}^{n}$ be n independent samples of $\bX$ and 
suppose $p = o(n/r^{*}(\bM_{\mathbb{E}(\bX|Y)}))$ where $r^{*}(\bM_{\mathbb{E}(\bX|Y)}) = \tr(\bM_{\mathbb{E}(\bX|Y)})/\lambda_{\max}(\bM_{\mathbb{E}(\bX|Y)})$. We have
\[
\| \widehat{\bM}_m - \bM_{\mathbb{E}(\bX|Y)} \|_2 = O_p((r^{*}(\bM_{\mathbb{E}(\bX|Y)}))^{1/2} (\log{p})^{1/2} n^{-1/2}) + O_p( H^{-\vartheta}) + O_p(H^{1/2}p^{-1/2}n^{-1/2}).
\]
\end{theorem}

The second part of Condition (1) is similar to that of \cite{HsingCarroll1992}, \cite{ZhuNg1995} and \cite{ZhuMiaoPeng2006}, which only requires $m \ge 2$ rather than $m \ge 4$ in their articles. The reason why we have a much milder condition here may originate from the first part of this condition, \ie, $\bX$ is restricted to be elliptically distributed. Condition (2) is similar to Condition (A3) of \cite{Lin2017}, where they imposed the boundary condition on the eigenvalues of $\bSigma$ while we assume the boundary property holds for $\bM_{\mathbb{E}(\bX|Y)}$. This extension is reasonable due to the connection between $\mathbb{E}(\bX|Y)$ and $\bX$. The last condition is exactly the same as the second part of Condition (A4) in \cite{Lin2017} and we put a remark (Remark \ref{r2}) in Section \ref{s7} for this condition.

Theorem \ref{t2} shows that the effect of $H$ on the convergence rate is two-sided. Hence, we'd better choose a moderate size for the number of the slices, not too big or too small. What's more, the growing rate of $p$ with respect to $n$ is roughly $p \log p = o(n)$, noting that $r^{*}(\bM_{\mathbb{E}(\bX|Y)}) \le p$. Therefore, we only require $H$ to diverge to infinity with $n$ and $p$ satisfying $p \log p = o(n)$ to make $\widehat{\bM}_m$ an consistent estimator of $\bM_{\mathbb{E}(\bX|Y)}$.

\begin{cy} \label{c2}
Under conditions of Theorem \ref{t2}, if the eigenvalues of $\bM$ is bounded away from zero and infinity and $p\log p = o(n^{1/2})$, with probability converging to $1$ we obtain
\[
\|\hat\bM^{-1} \hat\bM_{\bm} - \bM^{-1} \bM_{\mathbb{E}(\bX|Y)} \|_2 \to 0, \quad as \quad n \to \infty.
\]
\end{cy}

For any matrix $\bB$, denote $col(\bB)$ as the space spanned by the columns of $\bB$. Letting $\bV = (\bbeta_1^{\T}, \dots, \bbeta_K^{\T})$, from Theorem \ref{t1} we have 
\[
col(\bV) = \bM^{-1} col(\bM_{\mathbb{E}(\bX|Y)}).
\]
Let $\widehat\bV = \hat\bM^{-1} \hat\bM_{\bm}$, Corollary \ref{c2} implies that $\| \bP_{\widehat\bV} - \bP_{\bV} \|_2 \to 0$ which means the consistency of the ESIR method.

\setcounter{equation}{0}
\section{Numerical examples}
\label{s5}

We report several numerical results in this section, including different model designs for simulation data and a real data analysis of the Istanbul stock exchange data. The squared multiple correlation coefficient $R^2 (\hat\bbeta_i)$ (\cite{Li1991}, \cite{ZhuMiaoPeng2006}) is used to measure the distance between the ESIR estimator $\hat\bbeta_i$ and the central subspace $\bB$ for $i = 1, \dots, K$ and their average $R^2$ to measure the distance between the space formed by all the $\hat\bbeta$'s and the central subspace (\cite{Li1991}, \cite{ZhuMiaoPeng2006}). $R^2 (\bb)$ is calculated by
\[
R^2 (\bb) = \max_{\bbeta \in \bB} \frac{(\bb\bSigma\bbeta^{\T})^2}{\bb\bSigma\bb^{\T} \cdot \bbeta\bSigma\bbeta^{\T}}
\]
for any $1 \times p$ vector $\bb$. Thus, a bigger squared multiple correlation coefficient indicates more estimation efficiency.

\subsection{Single index model}

Three types of single index models are considered under multivariate normal distribution and other five frequently used elliptical distributions in this part, including the multivariate Laplace distribution, multivariate symmetric logistic distribution, multivariate Student' t distribution with degrees of freedom $2$ and $3$ and the multivariate Cauchy distribution.

Model (A1):
\[
Y = \frac{1}{0.5+(\bbeta_1 \bX+1.5)^2} +\sigma \bepsilonbb.
\]

Model (A2):
\[
Y = 0.5+(\bbeta_1 \bX+1.5)^2 + \sigma \bepsilonbb.
\]

Model (A3):
\[
Y = (\bbeta_1\bX+2) \cdot \sigma \bepsilonbb.
\]

Model (A1) and (A2) come from \cite{Li1991} and Model (A3) is stimulated by Example 3 of \cite{ZhuMiaoPeng2006}. In all the above three models, $\sigma = 0.5$, $\bbeta_1 = (1, 0, \dots, 0)$, 
$\bepsilonbb \sim N(0, 1)$ and $\bX \sim EC_p(\mathbf{0}, \bSigma, \xi)$ where $\bSigma = I_{p \times p}$ and $\xi$ is the generating variable. We change the distribution of $\bX$ among the elliptical distributions mentioned above by adjusting the distribution of the generating variable $\xi$. For multivariate logistic distribution, we choose the dependence parameter to be $0.2$ to indicate weak dependence among elements of $\bX$. The sample size $n$, the number of predictors $p$ and the number of the slices $H$ are chosen to be $400$, $10$ and $10$ respectively. Table \ref{tab1} reports the means and standard deviations (in parentheses) of $R^2(\hat\bbeta_1)$ after $100$ replicates under different simulation schemes.

\begin{table}[t!] 
\caption{Mean and standard deviation of $ R^2(\hat\bbeta_1)$ for the single index models}
\label{tab1}\par
\vskip .2cm
\centerline{\tabcolsep=3truept\begin{tabular}{lcccccc} \hline 
Distr of $\bX$&normal&Laplace&logistic&t ($3$)&t ($2$)&Cauchy\\ \hline
&$R^2(\hat\bbeta_1)$&$R^2(\hat\bbeta_1)$&$R^2(\hat\bbeta_1)$&$R^2(\hat\bbeta_1)$&$R^2(\hat\bbeta_1)$&$R^2(\hat\bbeta_1)$ \\ \hline
\multicolumn{7}{c}{Model (A1)}\\ \hline
SIR&0.95&0.88&0.90&0.71&0.37&0.10 \\
&(0.02)&(0.07)&(0.05)&(0.19)&(0.25)&(0.12) \\
ESIR&0.95&0.88&0.97&0.84&0.60&0.47 \\
&(0.02)&(0.08)&(0.00)&(0.09)&(0.30)&(0.33) \\ \hline
\multicolumn{7}{c}{Model (A2)}\\ \hline
SIR&1.00&0.97&1.00&0.77&0.42&0.18 \\
&(0.00)&(0.02)&(0.00)&(0.23)&(0.28)&(0.16) \\
ESIR&1.00&0.98&0.97&0.93&0.81&0.48 \\
&(0.00)&(0.01)&(0.00)&(0.05)&(0.18)&(0.36) \\ \hline
\multicolumn{7}{c}{Model (A3)}\\ \hline
SIR&0.91&0.82&0.90&0.66&0.34&0.16 \\
&(0.04)&(0.12)&(0.06)&(0.25)&(0.27)&(0.15) \\
ESIR&0.90&0.80&0.97&0.73&0.49&0.40 \\
&(0.05)&(0.13)&(0.01)&(0.20)&(0.30)&(0.34) \\ \hline
\end{tabular}}
\end{table}

As can be seen from Table \ref{tab1}, the ESIR method outperforms the SIR method in almost all the simulation schemes. Furthermore, the efficiency gain is more significant when the tail of the distribution of $\bX$ tends to become heavier, which can be easily seen from the simulation results for $t(3)$, $t(2)$ and Cauchy ($t(1)$) distributed covariates where a smaller degree of Student's t distribution indicates a heavier tail. For multivariate normal distribution and Laplace distribution, our ESIR estimator performs nearly as well as the SIR. Of note, the tail of the Laplace distribution is very close to that of normal distribution. See Figure 1 for the tails of the marginal distributions of $\bX$.

\begin{figure}\label{f0}
\centering
\includegraphics[scale=0.5]{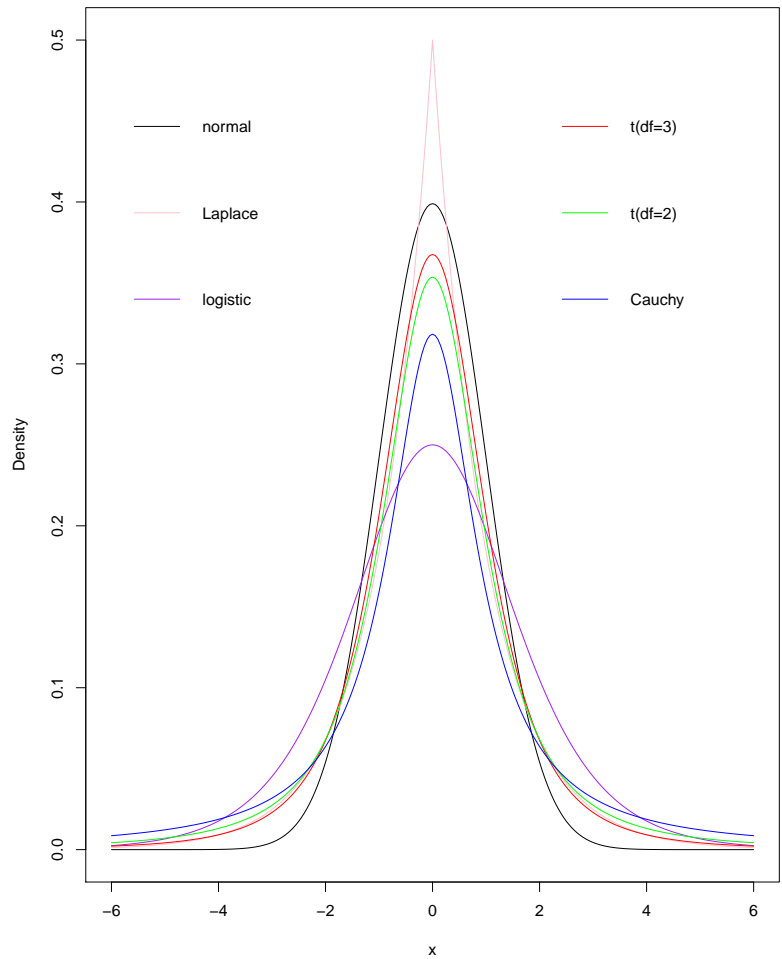}
\caption{Tails of elliptical distributions}
\end{figure}

\subsection{Other models}

Four models are considered in this section under $K = 2$ case for six different elliptical distributions of $\bX$. Unless otherwise noted, the simulation parameters used here are the same as those used in the first part for the single index model.

Model (B1):
\[
Y = \frac{\bbeta_1 \bX}{0.5+(\bbeta_2 \bX+1.5)^2} +\sigma \bepsilonbb,
\]
where $\bbeta_1 = (1, 0, \dots, 0)$, $\bbeta_2 = (0, 1, 0, \dots, 0)$, $\bepsilonbb \sim N(0, 1)$ and $\bX \sim EC_p(\mathbf{0}, \bSigma, \xi)$ where $\bSigma = I_{p \times p}$. This model was used by \cite{Li1991}.

Model (B2):
\[
Y = 4 + \bbeta_1\bX + (\bbeta_2\bX+2)\cdot \sigma \bepsilonbb.
\]
Here, we reset $p = 5$, $\bX \sim EC_p(\mathbf{0}, \bSigma, \xi)$ with $\bSigma = \diag\{2, 2, 2, 4, 2\}$, $\bbeta_{1} = (1,0,0,0,0)$ and $\bbeta_{2} = (0, 1, 1, 0, 0)$.

Model (B3):
\[
Y = (4 + \bbeta_1\bX) \cdot (\bbeta_2\bX+2)+ \sigma \bepsilonbb,
\]
where the simulation parameters for this model are the same as model (B2). Both of Model (B2) and (B3) stem from Example 2 and 3 of \cite{ZhuMiaoPeng2006}.

Model (B4):
\[
Y = (\bbeta_1\bX)^2 + |\bbeta_2\bX| + \sigma \bepsilonbb,
\]
where $\bbeta_1 = (0.5, 0.5, 0.5, 0.5, 0, \dots, 0)$ and $\bbeta_2 = (0.5, -0.5, 0.5, -0.5, 0, \dots, 0)$. This model stems from Example 3 of \cite{ChenCookZou2015}. The distribution of $\bX$ deviates a little bit from the elliptical distribution. That is, let $\bX = (X_1, \bX_2)$ where $\bX_2 = (X_2, \dots, X_p)$, $\bX_2 \sim EC_{p-1}(\mathbf{0}, \bSigma, \xi)$ where $\bSigma = (\sigma_{ij})$ with $\sigma_{ij} = 0.5^{|i-j|}$ for $i, j = 1, \dots, (p-1)$ and $X_1 = |X_2 +X_3 | + \bzeta$ where $\bzeta \sim N(0,1)$.

\begin{table}[t!] 
\caption{Mean and standard deviation of $ R^2(\hat\bbeta)$ for the other models I}
\label{tab3}\par
\vskip .2cm
\centerline{\tabcolsep=3truept\begin{tabular}{lccccccccc} \hline 
Distr of $\bX$&\multicolumn{3}{c}{normal}&\multicolumn{3}{c}{logistic}&\multicolumn{3}{c}{EC1} \\ \hline
&$R^2(\hat\bbeta_1)$&$R^2(\hat\bbeta_2)$&$R^2$&$R^2(\hat\bbeta_1)$&$R^2(\hat\bbeta_2)$&$R^2$&$R^2(\hat\bbeta_1)$
&$R^2(\hat\bbeta_2)$&$R^2$ \\ \hline
\multicolumn{10}{c}{Model (B1)}\\ \hline
SIR&0.96&0.88&0.92&0.91&0.20&0.56&0.22&0.19&0.21 \\
&(0.02)&(0.06)&&(0.04)&(0.18)&&(0.18)&(0.15)& \\
ESIR&0.94&0.76&0.85&0.99&0.18&0.59&0.89&0.84&0.87 \\
&(0.03)&(0.16)&&(0.00)&(0.17)&&(0.20)&(0.24)& \\ \hline
\multicolumn{10}{c}{Model (B2)}\\ \hline
SIR&0.99&0.81&0.90&0.99&0.71&0.85&0.48&0.44&0.46 \\
&(0.01)&(0.22)&&(0.01)&(0.27)&&(0.23)&(0.26)& \\
ESIR&0.99&0.62&0.81&1.00&0.74&0.87&0.94&0.85&0.90 \\
&(0.01)&(0.31)&&(0.00)&(0.24)&&(0.13)&(0.24)& \\ \hline
\multicolumn{10}{c}{Model (B3)}\\ \hline
SIR&1.00&0.93&0.97&1.00&0.26&0.63&0.38&0.42&0.40 \\
&(0.00)&(0.05)&&(0.00)&(0.26)&&(0.21)&(0.29)& \\
ESIR&1.00&0.78&0.89&1.00&0.23&0.62&0.88&0.87&0.88 \\
&(0.00)&(0.20)&&(0.00)&(0.24)&&(0.23)&(0.20)& \\ \hline
\multicolumn{10}{c}{Model (B4)}\\ \hline
SIR&0.97&0.66&0.82&1.00&0.67&0.84&0.43&0.42&0.43 \\
&(0.01)&(0.20)&&(0.00)&(0.21)&&(0.15)&(0.07)& \\
ESIR&0.97&0.69&0.83&1.00&0.90&0.95&0.92&0.81&0.87 \\
&(0.01)&(0.16)&&(0.00)&(0.03)&&(0.17)&(0.28)& \\ \hline
\end{tabular}}
\end{table}

\begin{table}[t!] 
\caption{Mean and standard deviation of $ R^2(\hat\bbeta)$ for other models II}
\label{tab4}\par
\vskip .2cm
\centerline{\tabcolsep=3truept\begin{tabular}{lccccccccc} \hline 
Distr of $\bX$&\multicolumn{3}{c}{t($3$)}&\multicolumn{3}{c}{t($2$)}&\multicolumn{3}{c}{Cauchy (t($1$))}\\ \hline
&$R^2(\hat\bbeta_1)$&$R^2(\hat\bbeta_2)$&$R^2$&$R^2(\hat\bbeta_1)$&$R^2(\hat\bbeta_2)$&$R^2$&$R^2(\hat\bbeta_1)$
&$R^2(\hat\bbeta_2)$&$R^2$ \\ \hline
\multicolumn{10}{c}{Model (B1)}\\ \hline
SIR&0.89&0.72&0.81&0.77&0.48&0.63&0.25&0.21&0.23 \\
&(0.10)&(0.15)&&(0.16)&(0.20)&&(0.21)&(0.15)& \\
ESIR&0.93&0.54&0.74&0.90&0.49&0.70&0.79&0.66&0.73 \\
&(0.05)&(0.24)&&(0.07)&(0.27)&&(0.26)&(0.28)& \\ \hline
\multicolumn{10}{c}{Model (B2)}\\ \hline
SIR&0.98&0.42&0.70&0.90&0.36&0.63&0.67&0.33&0.50 \\
&(0.02)&(0.32)&&(0.13)&(0.27)&&(0.23)&(0.22)& \\
ESIR&0.98&0.32&0.65&0.98&0.37&0.68&0.95&0.67&0.81 \\
&(0.02)&(0.28)&&(0.03)&(0.30)&&(0.09)&(0.33)& \\ \hline
\multicolumn{10}{c}{Model (B3)}\\ \hline
SIR&0.96&0.73&0.85&0.85&0.47&0.66&0.40&0.41&0.41 \\
&(0.05)&(0.25)&&(0.15)&(0.28)&&(0.27)&(0.27)& \\
ESIR&0.99&0.59&0.79&0.95&0.47&0.71&0.86&0.68&0.77 \\
&(0.01)&(0.30)&&(0.06)&(0.31)&&(0.20)&(0.28)& \\ \hline
\multicolumn{10}{c}{Model (B4)}\\ \hline
SIR&0.93&0.33&0.63&0.85&0.22&0.54&0.62&0.12&0.37 \\
&(0.07)&(0.24)&&(0.14)&(0.20)&&(0.19)&(0.13)& \\
ESIR&0.91&0.40&0.66&0.87&0.34&0.61&0.86&0.56&0.71 \\
&(0.05)&(0.25)&&(0.10)&(0.25)&&(0.15)&(0.32)& \\ \hline
\end{tabular}}
\end{table}

In these $K = 2$ cases, we include another elliptical distribution $EC1$ and remove the Laplace distribution, because the result of Laplace distribution is quite similar to that of normal distribution. Define $EC1 = EC_p(\mathbf{0}, \bSigma, \xi_1)$ with $\xi_1 \sim F(p, 1)$ where F indicates F distribution. Here $\xi_1$ does not have finite mean. 
This distribution was also used in \cite{HanLiu2016}. Table \ref{tab3} and \ref{tab4} exhibit a little bit difference from that of single index models. That is, while the first leading eigenvector or direction presents almost the same efficiency improvement as in the single index case,  the second estimated direction performs not so well as the SIR estimator under several simulation settings. However, one can find that when the tail of the distribution gets heavier, the ESIR estimation for the second e.d.r. direction inclines to become more accurate. From a comprehensive point of view, the  ESIR estimation efficiency is comparable or better than that of the SIR method. The robustness of ESIR is well demonstrated in Table \ref{tab4}, \ie, when the tail of the distribution of the covariates goes heavier (from t(3) to t(1)), the performance of our proposed ESIR method is getting better.

To examine the influence of $p$, $H$ and $n$ on the estimation efficiency of the ESIR estimator, we consider the combinations of
$n = 120, 200, 400$, $p = 5, 10, 30$ and $H = 5, 10, 20, 40$ in Model (B1) for Cauchy distributed covariates. Simulation results are presented in Table \ref{tab2} after $100$ replicates.

\begin{table}[t!] 
\caption{Estimation of the central subspace for Model (B1)}
\label{tab2}\par
\vskip .2cm
\centerline{\tabcolsep=3truept\begin{tabular}{clrrr|rrr|rrr|rrr} \hline 
\multicolumn{2}{c}{H}&\multicolumn{3}{c}{5}&\multicolumn{3}{c}{10}&\multicolumn{3}{c}{20}&\multicolumn{3}{c}{40} \\ \hline
\multicolumn{2}{c}{p}&5&10&30&5&10&30&5&10&30&5&10&30\\ \hline
\multicolumn{14}{c}{$n=120$}\\ \hline
SIR&$R^2(\hat\bbeta_1)$&0.47&0.24&0.12&0.47&0.27&0.12&0.46&0.24&0.10&0.41&0.24&0.09\\
&$R^2(\hat\bbeta_2)$&0.42&0.21&0.08&0.41&0.22&0.07&0.42&0.24&0.08&0.41&0.17&0.08\\
ESIR&$R^2(\hat\bbeta_1)$&0.86&0.82&0.74&0.83&0.77&0.73&0.87&0.76&0.69&0.82&0.74&0.67\\
&$R^2(\hat\bbeta_2)$&0.62&0.60&0.56&0.68&0.60&0.59&0.70&0.66&0.54&0.75&0.64&0.58\\ \hline
\multicolumn{14}{c}{$n=200$}\\ \hline
SIR&$R^2(\hat\bbeta_1)$&0.52&0.29&0.13&0.44&0.28&0.13&0.48&0.32&0.10&0.43&0.25&0.10\\
&$R^2(\hat\bbeta_2)$&0.37&0.25&0.08&0.44&0.21&0.09&0.43&0.21&0.07&0.40&0.17&0.09\\
ESIR&$R^2(\hat\bbeta_1)$&0.87&0.78&0.70&0.86&0.75&0.63&0.85&0.79&0.72&0.87&0.74&0.68\\
&$R^2(\hat\bbeta_2)$&0.66&0.56&0.49&0.67&0.58&0.57&0.66&0.62&0.54&0.65&0.66&0.56\\ \hline
\multicolumn{14}{c}{$n=400$}\\ \hline
SIR&$R^2(\hat\bbeta_1)$&0.48&0.29&0.11&0.48&0.28&0.10&0.48&0.26&0.12&0.43&0.25&0.09\\
&$R^2(\hat\bbeta_2)$&0.40&0.21&0.08&0.42&0.22&0.06&0.39&0.19&0.07&0.39&0.20&0.07\\
ESIR&$R^2(\hat\bbeta_1)$&0.87&0.80&0.76&0.91&0.73&0.72&0.92&0.86&0.67&0.88&0.77&0.67\\
&$R^2(\hat\bbeta_2)$&0.68&0.58&0.56&0.70&0.60&0.64&0.71&0.64&0.58&0.71&0.65&0.63\\ \hline
\end{tabular}}
\end{table}

In this setting, we avoid reporting the stand deviations and the averages of $\hat\bbeta_1$ and $\hat\bbeta_2$ to improve the clarity of the simulation results. From Table \ref{tab2}, we find that when $n$ and $H$ stay fixed, the larger $p$ causes $R^2(\hat\bbeta_1)$ and $R^2(\hat\bbeta_2)$ to become smaller as high dimension reduces the estimation efficiency of both SIR and ESIR. However, our ESIR method is not so sensitive to dimensionality as SIR. Looking at the rows of Table \ref{tab2}, $R^2(\hat\bbeta_1)$ and $R^2(\hat\bbeta_2)$ regarding the ESIR method decrease much slower when $p$ gets larger. Secondly, when $n$ gets larger, both SIR and ESIR tend to perform better which fits our expectation. Lastly, the number of slices doesn't have any significant impact on the estimation of both methods. It's not surprising because \cite{ZhuNg1995} and \cite{ZhuMiaoPeng2006} also found such a phenomenon in their simulation studies for SIR method. We don't present simulation results for other distributions or other models here as they are quite similar to those of Table \ref{tab2}.

\subsection{Real data analysis}

In this part, we exploit the Istanbul stock exchange data set \\ (\href{http://archive.ics.uci.edu/ml/datasets/ISTANBUL+STOCK+EXCHANGE}{http://archive.ics.uci.edu/ml/datasets/ISTANBUL+STOCK+EXCHANGE}) to demonstrate the superiority of ESIR in contrast with SIR when the explaining variables are characterized by non-Gaussian and heavy-tailed features. 
There are $8$ variables in the data sets from Jan 5, 2009 to Feb 22, 2011 (536 rows in all):  Istanbul stock exchange national 100 index (ISE), Standard Poole 500 return index (SP), Stock market return index of Germany (DAX), Stock market return index of UK (FTSE), Stock market return index of Japan (NIKKEI), Stock market return index of Brazil (BOVESPA), MSCI European index (EU) and MSCI emerging markets index (EM). We choose EM as the response variable and the other variables as the covariates to formulate a regression problem.

Firstly, we explore the marginal distributions of the independent variables. Two normality tests, Shapiro-Wilk test and Kolmogorov-Smirnovare test, are conducted to check the non-Gaussian feature of these variables. Table \ref{tab5} summarizes the Shapiro-Wilk statistics and Kolmogorov-Smirnovare statistics. We find that the covariates are all considered as non-Gaussian distributed at the significance level of $0.05$ except that the conclusion on the variable ISE is questionable. We then plot the empirical densities of the standardized covariates against the stand normal distribution to illustrate the heavy-tailed pattern. Figure 2 exhibits this heavy-tailed character pretty clearly. It can also be seen from this figure that all the independent variables are symmetric about $0$. Thus, we can readily apply our ESIR method to this data set.

\begin{table}[t!] 
\caption{Normality tests. The '***', '*' and '$\cdot$' in cells represents the p-value less than $0.001$, $0.05$ and $0.15$ respectively.}
\label{tab5}\par
\vskip .2cm
\centerline{\tabcolsep=3truept\begin{tabular}{lcc} \hline 
&Shapiro-Wilk&Kolmogorov-Smirnov \\ \hline
ISE&$0.98^{***}$&$0.05^{\cdot}$\\
SP&$0.94^{***}$&$0.11^{***}$\\
DAX&$0.97^{***}$&$0.07^{*}$\\
FTSE&$0.97^{***}$&$0.07^{*}$\\
NIKKEI&$0.98^{***}$&$0.07^{*}$\\
BOVESPA&$0.97^{***}$&$0.06^{*}$\\
EU&$0.97^{***}$&$0.07^{*}$\\ \hline
\end{tabular}}
\end{table}

\begin{figure}\label{fig2}
\centering
\includegraphics[scale=0.5]{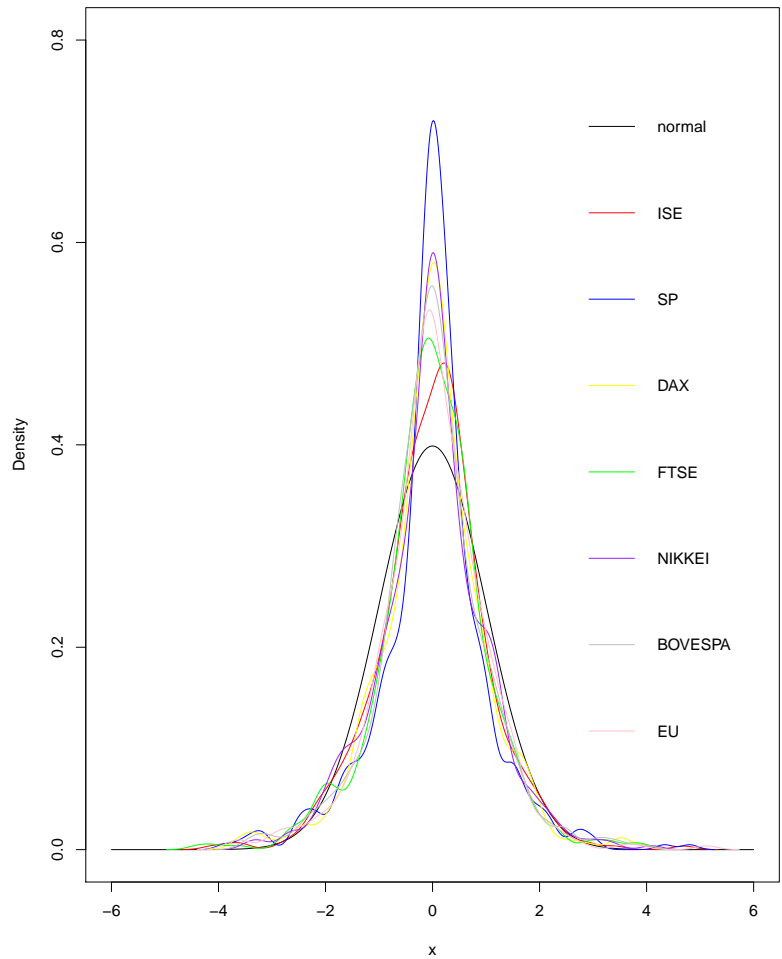}
\caption{Empirical densities of the covariates}
\end{figure}

To determine the dimension $K$ of the central subspace, we apply the widely-used marginal dimension test with the number of slices being $10$. The test result suggests that $K=2$ would be a proper choice. Therefore, we set the dimension of the central subspace to be $K=2$ and the number of slices $H = 10$. 
After estimating the e.d.r. directions, we get two new factors: $\hat\bbeta_1\bX$ and $\hat\bbeta_2\bX$, then use them and $(\hat\bbeta_1\bX)^2$, $(\hat\bbeta_2\bX)^2$, and $(\hat\bbeta_1\bX) \cdot (\hat\bbeta_2\bX)$ as explanatory variables to fit $EM$. The adjusted R-squared and F statistics are then exploited to compare the performances of ESIR and SIR. We first try samples of the whole time period (in fact, we used the first $500$ samples for computational convenience) and extend to investigate three shorter periods which simultaneously appear to possess significant heavy tails (see Figure 3). The results are presented in Table \ref{tab6}. Obviously, our method outperforms the SIR in all the four periods with significantly larger values of both R-squared and F statistics. This finding is complied with the simulation results above. It can be conjectured that the ESIR method would work better for individual asset returns and more risky financial derivatives.

\begin{figure}\label{fig3}
\centering
\includegraphics[width=0.8\textwidth,height=0.8\textheight]{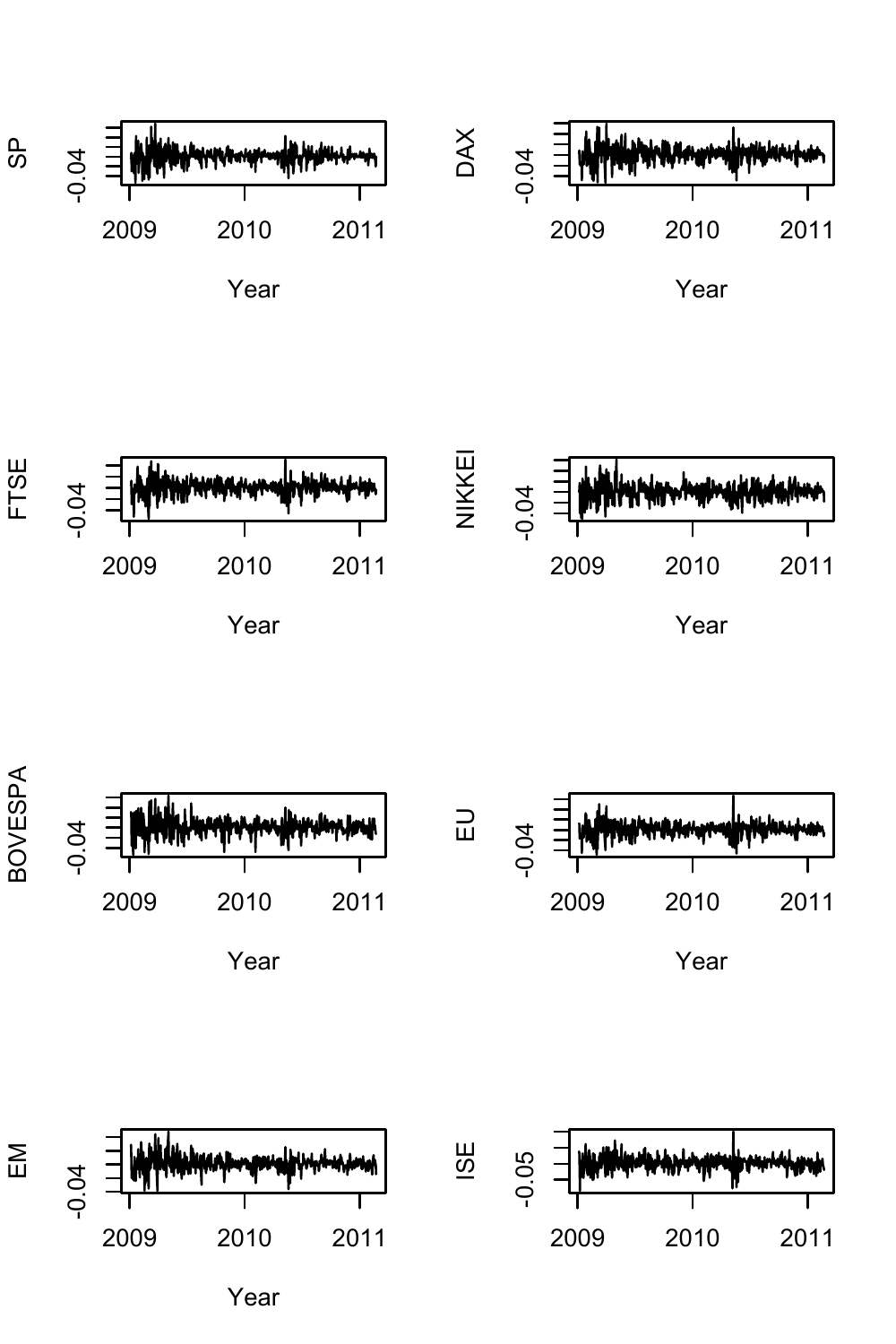}
\caption{The variables}
\end{figure}

\begin{table}[t!] 
\caption{Regression results.}
\label{tab6}\par
\vskip .2cm
\centerline{\tabcolsep=3truept\begin{tabular}{lcccc} \hline 
Time period&\multicolumn{2}{c}{Adjusted $R^2$}&\multicolumn{2}{c}{F statistic} \\ \hline
&SIR&ESIR&SIR&ESIR\\ \hline
2009.01.05-2009.03.13&0.69&0.71&241.30&264.40 \\
2009.03.30-2009.06.10&0.60&0.73&159.60&291.90 \\
2010.04.27-2010.07.06&0.64&0.71&187.70&266.30 \\
2009.01.05-2011.01.04&0.56&0.71&137.10&267.60 \\ \hline
\end{tabular}}
\end{table}

\setcounter{equation}{0}
\section{Discussion}
\label{s6}

In this paper, we propose the elliptically sliced inverse regression method for sufficient dimension reduction, which is a robust alternative to SIR for analysing high dimensional, elliptically distributed data. The main idea is to introduce the multivariate Kendall's tau matrix in a generalized eigenvector problem to cope with the heavy-tailed elliptically distributed covariates. We then present the main theorem to demonstrate the rationality of the ESIR estimator and give a practical algorithm for the ESIR method. The asymptotic behavior of the ESIR estimator is studied and the corresponding convergence rate is obtained in high dimensional setting. Simulation results demonstrate that ESIR significantly improves the estimation efficiency for the central subspace under the setting of elliptically distributed covariates. Moreover, our method can be easily extended to most other sufficient dimension reduction methods such as SAVE, directional reduction (DR, \cite{LiWang2007}) and principal fitted components (PFC, \cite{CookForzani2009}) etc. Please refers to \cite{Li2007} and \cite{ChenZouCook2010} for generalized eigenvector problem. Lastly, our method is of vital importance for analyzing heavy-tailed financial, genomic and bioimaging data. Of note, we do not spend energy on ascertaining the dimension of the central subspace in this paper and leave it for further study.


\setcounter{equation}{0}
\section{Proofs}
\label{s7}

\subsection{Lemma 1}

\begin{lemma}\label{l1}
Assume $\bX \sim EC_p(\bmu, \bSigma, \xi)$ and $\cov(\mathbb{E}(\bX|Y)) = \bOmega \bLambda \bOmega^{\T}$, where $\bOmega = (\bomega^{(1)}, \dots, \bomega^{(p)})^{\T}$ is the $p \times p$ matrix of the eigenvectors and $\bLambda = \diag(\lambda_1, \dots, \lambda_p)$ with $\lambda_1 \ge \lambda_2 \ge \dots \ge \lambda_p$ being the corresponding eigenvalues. Letting $\bM_{\mathbb{E}(\bX|Y)}$ denote the population Kendall's tau matrix of the vector $\mathbb{E}(\bX|Y)$, by Proposition of \cite{Marden1999} we have
\[
\bM_{\mathbb{E}(\bX|Y)} = \bOmega \bLambda_1 \bOmega^{\T},
\]
where $\bLambda_1$ is a $p \times p$ diagonal matrix containing the eigenvalues of $\bM_{\mathbb{E}(\bX|Y)}$.
\end{lemma}

\begin{proof}
The proof for this Lemma is mainly based on the results of \cite{Marden1999}. For any $\bX \sim EC_p(\bmu, \bSigma, \xi)$, we have the decomposition below:
\be \label{7.1}
\bX = \bOmega\bW + \bb
\ee
where $\bOmega$ is some orthogonal matrix, $\bW \in \mathbb{R}^p$ is coordinatewise symmetric about 0, that is
\be \label{7.2}
\bg \bW \overset{\rm{d}}{=} \bW
\ee
for any matrix $\bg $ with $\bg_{jj} \in \{1, -1\}$ and $\bg_{ij}=0 (i \ne j)$ and $\bb$ is some $p \times 1$ centering vector. We assume that $\cov (\bW)$ exists, and without loss of generality, that $\lambda_1 \ge \dots \ge \lambda_p$ with $\lambda_i = \var (W_i)$ for $\bW = (W_1, \dots, W_p)^{\T}$. Thus, we obtain $\bSigma = \cov (\bX) = \bOmega \bLambda \bOmega^{T}$, where $\bLambda = \diag (\lambda_1, \dots, \lambda_p)$ with $\lambda_1 \ge \dots \ge \lambda_p$.

For the vector $F(Y) = \mathbb{E} (\bX | Y)$, from (\ref{7.1}) we have
\be \label{7.3}
F(Y) = \mathbb{E} (\bX | Y) = \mathbb{E} (\bOmega \bW + \bb | Y) = \bOmega \mathbb{E} (\bW | Y) + \bb := \bOmega F_{\bW}(Y) + \bb.
\ee
Then for $F_{\bW}(Y) = \mathbb{E} (\bW | Y)$, we can derive from (\ref{7.2}) for any $\bg $ with $\bg_{jj} \in \{1, -1\}$ and $\bg_{ij}=0 (i \ne j)$,
\be \label{7.4}
\bg F_{\bW}(Y) = \bg \mathbb{E} (\bW | Y) = \mathbb{E} (\bg \bW | Y) = \mathbb{E} (\bW | Y) = F_{\bW}(Y).
\ee
Therefore, $F_{\bW}(Y) = \mathbb{E} (\bW | Y)$ is coordinatewise symmetric about 0.

By Proposition of \cite{Marden1999} and (\ref{7.3}) and (\ref{7.4}), we obtain
\[
\bM_{\mathbb{E} (\bX | Y)} = \bOmega \bLambda_1 \bOmega^{\T},
\]
where $\bLambda_1$ is a $p \times p$ diagonal matrix whose diagonal elements are the eigenvalues of $\bM_{\mathbb{E} (\bX | Y)}$.
\end{proof}

\subsection{Proof of Theorem \ref{t2}}\label{pt1}

\begin{proof}
\[
\begin{split}
&~ \| \widehat{\bM}_\bm - \bM_{\mathbb{E}(\bX|Y)} \|_2\\
&~ = \| \frac{2}{n(n-1)} \sum_{i' < i} \frac{(\bm_i-\bm_{i'})(\bm_i-\bm_{i'})^{\T}}{\| \bm_i-\bm_{i'} \|_2^2} - \mathbb{E} \frac{(\bm(Y)-\bm(\widetilde Y))(\bm(Y)-\bm(\widetilde Y))^{\T}}{\| \bm(Y)-\bm(\widetilde Y) \|_2^2}\\
&~ + \frac{2}{H(H-1)} \sum_{h'<h} \frac{(\bm_h-\bm_{h'})(\bm_h-\bm_{h'})^{\T}}{\|\bm_h-\bm_{h'}\|_2^2} - \frac{2}{n(n-1)} \sum_{i' < i} \frac{(\bm_i-\bm_{i'})(\bm_i-\bm_{i'})^{\T}}{\| \bm_i-\bm_{i'} \|_2^2} \\
&~ + \frac{2}{H(H-1)} \sum_{h'<h} \frac{(\hat\bm_h-\hat\bm_{h'})(\hat\bm_h-\hat\bm_{h'})^{\T}}{\|\hat\bm_h-\hat\bm_{h'}\|_2^2} \|_2 - \frac{2}{H(H-1)} \sum_{h'<h} \frac{(\bm_h-\bm_{h'})(\bm_h-\bm_{h'})^{\T}}{\|\bm_h-\bm_{h'}\|_2^2} \|_2\\
&~ :=\| A_1 + A_2 + A_3 \|_2 \\
&~ \le \|A_1\|_2 + \|A_2\|_2 + \|A_3\|_2
\end{split}
\]
where $\bm_h = \mathbb{E}(\bX|Y \in I_h)$ and $\hat{\bm}_h = 1/l \sum_{\mathrm{y}_i \in I_h} \bX_i $ for $(h = 1, \dots, H)$.

For the first part $A_1$, by Theorem 3.1 of \cite{HanLiu2016}, if $p = o(n/r^{*}(\bM_{\mathbb{E}(\bX|Y)}))$, we can easily obtain
\[
\|A_1\|_2 = O_p \bigg(\|\bM_{\mathbb{E}(\bX|Y)} \|_2 \sqrt{r^{*}(\bM_{\mathbb{E}(\bX|Y)}) \log{p}/n}\bigg) = O_p\bigg(\sqrt{r^{*}(\bM_{\mathbb{E}(\bX|Y)}) \log{p}/n}\bigg),
\]
where $r^{*}(\bM_{\mathbb{E}(\bX|Y)}) = \frac{\tr(\bM_{\mathbb{E}(\bX|Y)})}{\lambda_1(\bM_{\mathbb{E}(\bX|Y)})} $. The last equality is ensured by Condition (2).

For the second part $A_2$, if Condition (3) stands, from Remark 1 and 2 in \cite{Lin2017}, we get
\be \label{7.5}
\|\frac{1}{H} \sum_{h=1}^H \bm_h \bm_h^{\T} - \frac{1}{n} \sum_{i=1}^n \bm_i \bm_i^{\T} \|_2= O_p(H^{-\vartheta}).
\ee
Using the relationship of $\hat\bSigma$ and $\hat\bM$ in Section \ref{s2} and the second part of Condition (1), we obtain
\[
\|A_2\|_2 = O_p( H^{-\vartheta}).
\]

For the third part $A_3$, elementary probability theory tells us that the elements of $\hat\bm_h$ converge to those of $\bm_h$ at the rate of $l^{-1/2}$. Thus, $\|\hat\bm_h - \bm_h\|_2 = O_p(p^{1/2}l^{-1/2})$ and $\| \hat\bm_h - \hat\bm_{h'} - (\bm_h - \bm_{h'}) \|_2 = O_p(p^{1/2}l^{-1/2}) $. Now we introduce some notation for easy presentation: $\widehat M_{hh'} := \hat\bm_h - \hat\bm_{h'}$ and $ M_{hh'} := \bm_h - \bm_{h'}$. Then we only need to consider the term below:
\[
\begin{split}
&~ \| \frac{\widehat M_{hh'} \widehat M_{hh'}^{\T} }{\| \widehat M_{hh'} \|_2^2 } - \frac{ M_{hh'} M_{hh'}^{\T} }{\| M_{hh'} \|_2^2 } \|_2\\
&~ = \| \frac{\widehat M_{hh'} \widehat M_{hh'}^{\T} }{\| \widehat M_{hh'} \|_2^2 } - \frac{\widehat M_{hh'} \widehat M_{hh'}^{\T} }{\| M_{hh'} \|_2^2 } + \frac{\widehat M_{hh'} \widehat M_{hh'}^{\T} }{\| M_{hh'} \|_2^2 } - \frac{ M_{hh'} M_{hh'}^{\T} }{\| M_{hh'} \|_2^2} \|_2\\
&~ \le \| \frac{\widehat M_{hh'} \widehat M_{hh'}^{\T} }{\| \widehat M_{hh'} \|_2^2 } - \frac{\widehat M_{hh'} \widehat M_{hh'}^{\T} }{\| M_{hh'} \|_2^2 } \|_2 + \| \frac{\widehat M_{hh'} \widehat M_{hh'}^{\T} }{\| M_{hh'} \|_2^2 } - \frac{ M_{hh'} M_{hh'}^{\T} }{\| M_{hh'} \|_2^2} \|_2\\
&~ \le \| \widehat M_{hh'} \widehat M_{hh'}^{\T} \|_2 | \frac{\|M_{hh'}\|_2^2 - \|\widehat M_{hh'}\|_2^2 }{\|\widehat M_{hh'}\|_2^2 \|M_{hh'}\|_2^2 } | + \frac{1}{\|M_{hh'}\|_2^2} \{ \| (\widehat M_{hh'} - M_{hh'})\widehat M_{hh'}^{\T} \|_2 + \| M_{hh'} (\widehat M_{hh'} - M_{hh'})^{\T} \|_2 \} \\
&~ \le \frac{\| \widehat M_{hh'} \widehat M_{hh'}^{\T} \|_2 (\|M_{hh'}\|_2 + \| \widehat M_{hh'} \|_2 )}{\|\widehat M_{hh'}\|_2^2 \|M_{hh'}\|_2^2 } \|M_{hh'} - \widehat M_{hh'} \|_2  + \frac{\|M_{hh'}\|_2 + \| \widehat M_{hh'} \|_2 }{\|M_{hh'}\|_2^2} \|M_{hh'} - \widehat M_{hh'} \|_2
.
\end{split}
\]
Condition (1) implies that $|\hat\bm_h\|_2 = \| 1/l \sum_{1=1}^{l} \bX_{hi} \|_2 \le 1/l \sum_{1=1}^{l} \| \bX_{hi} \|_2 = O_p(p)$. Then, $\| \widehat M_{hh'} \|_2 = \| \hat\bm_h - \hat\bm_{h'} \|_2 \le \| \hat\bm_h \|_2 + \| \hat\bm_{h'} \|_2 = O_p(p)$ and $\| M_{hh'} \|_2 = \|M_{hh'} - \widehat M_{hh'} + \widehat M_{hh'} \|_2 \le \|M_{hh'} - \widehat M_{hh'} \|_2 + \|\widehat M_{hh'} \|_2 = O_p(p^{1/2}l^{-1/2}) + O_p(p) = O_p(p)$. Thus, $ \frac{\|M_{hh'}\|_2 + \| \widehat M_{hh'} \|_2 }{\|M_{hh'}\|_2^2} = O_p(p^{-1})$. In view of $\| \widehat M_{hh'} \widehat M_{hh'}^{\T} \|_2 /\|\widehat M_{hh'}\|_2^2 = O_p(1)$, we obtain
\[
\|A_3\|_2 = O_p(p^{-1}) \cdot O_p(p^{1/2} l^{-1/2}) = O_p(p^{-1/2}H^{1/2}n^{-1/2}).
\]



\end{proof}

\begin{remark}
Although \cite{HanLiu2016} required $\bX \sim EC_p(\bmu, \bSigma, \xi)$ in their Theorem 3.1, one can easily find that it is not a necessary condition by reviewing their proof, which means in our paper that we don't need any restriction on the distribution of $\bm(Y)$ for the consistency. That's a quite good property, because it is not trivial to test the distribution of $\bm(Y)$.
\end{remark}

\begin{remark} \label{r2}
The convergence rate for the second part does not seem quite apparent. In fact, there is an underling assumption. That is, we believe that the weights plugged into (\ref{2.1}) would not change the decreasing rate of (\ref{7.5}), which is a power of $H$. This speculation is tenable, because the weights take effects as the reciprocal of the $L_2$ distance. This assumption could be further demonstrated by the following finding:
\[
\begin{split}
~& \| \frac{\bm_h-\bm_{h'}}{\|\bm_h-\bm_{h'}\|_2} - \frac{\bm_i-\bm_{i'}}{\|\bm_i-\bm_{i'}\|_2} \|_2 \\
~& := \| \frac{\ba}{\|\ba\|_2} - \frac{\bb}{\|\bb\|_2} \|_2 \\
~& = \| \frac{\ba}{\|\ba\|_2} -\frac{\bb}{\|\ba\|_2} + \frac{\bb}{\|\ba\|_2}- \frac{\bb}{\|\bb\|_2} \|_2 \\
~& \le \| \frac{\ba}{\|\ba\|_2} -\frac{\bb}{\|\ba\|_2} \|_2 + \|  \frac{\bb}{\|\ba\|_2}- \frac{\bb}{\|\bb\|_2} \|_2 \\
~& \le \frac{\|\ba-\bb\|_2}{\|\ba\|_2} + \|\bb\|_2 \cdot \frac{| \|\bb\|_2-\|\ba\|_2}{\|\ba\|_2\|\bb\|_2}\\
~& \le 2 \cdot \frac{\|\ba-\bb\|_2}{\|\ba\|_2}\\
~& \le 2 \cdot \frac{\|\bm_h-\bm_{h'} - (\bm_i-\bm_{i'})\|_2}{\|\bm_h-\bm_{h'}\|_2}.
\end{split}
\]
\end{remark}


\subsection{Proof of Corollary \ref{c2}}

\begin{proof}
\[
\begin{split}
&~ \|\widehat\bM^{-1} \widehat\bM_{\bm} - \bM^{-1} \bM_{\mathbb{E}(\bX|Y)} \|_2 \\
&~ = \|\widehat\bM^{-1} \widehat\bM_{\bm} - \bM^{-1} \widehat\bM_{\bm} + \bM^{-1} \widehat\bM_{\bm} - \bM^{-1} \bM_{\mathbb{E}(\bX|Y)} \|_2 \\
&~ \le \|\widehat\bM^{-1} \widehat\bM_{\bm} - \bM^{-1} \widehat\bM_{\bm} \|_2 + \|\bM^{-1} \widehat\bM_{\bm} - \bM^{-1} \bM_{\mathbb{E}(\bX|Y)} \|_2 \\
&~ \le \| \widehat\bM^{-1}- \bM^{-1}\|_2 \| \widehat\bM_{\bm} \|_2 + \| \bM^{-1} \|_2 \| \widehat\bM_{\bm} - \bM_{\mathbb{E}(\bX|Y)} \|_2.
\end{split}
\]

From Theorem 3.2 of \cite{HanLiu2016} in combination with that the eigenvalues of $\bM$ are bounded away from zero and infinity, we have $\|\widehat\bM^{-1}- \bM^{-1}\|_2 = \| \bM^{-1} (\bM - \widehat \bM) \widehat\bM^{-1} \|_2 = O_p((p\log p)^{1/2}n^{-1/2})$, where the last equality can be demonstrated by
\[
\begin{split}
\| \widehat\bM^{-1} \|_2
~& = (\lambda_{\max}(\widehat\bM^{-1} \widehat\bM^{-1}))^{1/2}\\
~& = \lambda_{\max}(\widehat\bM^{-1}) = \lambda_{\min}(\widehat\bM) \\
~& = \min_{\|\balpha \|_2=1} \balpha^{\T} \widehat\bM \balpha \\
~& = \min_{\|\balpha \|_2=1} \balpha^{\T} (\widehat\bM-\bM+\bM) \balpha \\
~& \ge \min_{\|\balpha \|_2=1} \balpha^{\T}(\widehat\bM-\bM) \balpha + \min_{\|\balpha \|_2=1} \balpha^{\T}\bM\balpha \\
~& \ge \| \widehat\bM-\bM \|_2 + \lambda_{\min}(\bM) \\
~& = o_p(1) + \lambda_{\min}(\bM) \\
~& = O_p(1).
\end{split}
\]
Moreover, in view of Condition (2) in Theorem \ref{t2}, we can easily get $\| \widehat\bM_{\bm} \|_2 = O_p(1)$. Together with the result of Theorem \ref{t2} about $\| \widehat\bM_{\bm} - \bM_{\mathbb{E}(\bX|Y)} \|_2$, we complete the proof.

\end{proof}

\end{document}